\documentclass[12pt]{article}
\DeclareMathSizes{12}{12}{10}{10}

\usepackage[margin=1in]{geometry}

\usepackage{amsmath, mathtools, graphicx}
\usepackage{amsthm}
\usepackage{hyperref}
\usepackage[below]{placeins}

\usepackage{mathrsfs, appendix}
\usepackage{nicefrac}
\usepackage{xfrac}
\usepackage{xcolor}
\usepackage{physics}
\usepackage{graphicx}
\usepackage{subcaption}
\usepackage{rotating}
\makeatletter
\def\maxwidth#1{\ifdim\Gin@nat@width>#1 #1\else\Gin@nat@width\fi}
\makeatother

\makeatletter
\newcommand{\ignorecite}[1]{{\@fileswfalse\cite{#1}}}%
\makeatother

\newcommand*{\JJ}{\mathscr{J}}
\newcommand*{\GG}{\mathscr{G}}

\theoremstyle{definition}
\newtheorem{theorem}{Theorem}
\newtheorem{lemma}{Lemma}
\newtheorem{property}{Property}

\begin{document}

\title{Direct Interaction Approximation for generalized stochastic models in the turbulence problem}
\author{B.K. Shivamoggi and N. Tuovila \\ {\small University of Central Florida} \\ {\small Orlando, FL 32816-1364}}
\date{}
\maketitle

\begin{abstract}
    The purpose of this paper is to consider the application of the direct interaction approximation (DIA) developed by Kraichnan \ignorecite{Kraichnan-1961}, \ignorecite{Kraichnan-1965} to generalized stochastic models in the turbulence problem.  Previous developments (Kraichnan \ignorecite{Kraichnan-1961}, \ignorecite{Kraichnan-1965}, Shivamoggi et al. \ignorecite{BKS}, Shivamoggi and Tuovila \ignorecite{BT}) were based on the Boltzmann-Gibbs prescription for the underlying entropy measure, which exhibits the extensivity property and is suited for ergodic systems.  Here, we consider the introduction of an influence bias discriminating rare and frequent events explicitly, as it behooves non-ergodic systems, which is dealt with by a using a Tsallis type \ignorecite{Tsallis-intro} autocorrelation model with an underlying non-extensive entropy measure.  As an example, we consider a linear damped stochastic oscillator system, and describe the resulting stochastic process.  The non-perturbative aspects excluded by Keller's perturbative procedure \ignorecite{Keller} are found to be minimized in the white-noise limit.  In the opposite limit, the physical variances between the random process models don't seem to materialize, and the Uhlenbeck-Ornstein and Tsallis type models are found to yield the same result.  In the process, we also deduce some apparently novel mathematical properties of the stochastic models associated with the present investigation -- the gamma distribution (Andrews et al. \ignorecite{APS1989}, \ignorecite{AS1990}) and the Tsallis non-extensive entropy \ignorecite{Tsallis-intro}.
\end{abstract}

\section{Introduction}

The \emph{direct-interaction approximation} (DIA) advanced by Kraichnan \cite{Kraichnan-1961}, \cite{Kraichnan-1965}\footnote{Notwithstanding several important insights afforded by the DIA into turbulence dynamics, some issues still hamper the DIA -- one such issue is the unphysical effect of the large-scale driving mechanisms on the energy transfer in the inertial range (Kraichnan \cite{Kraichnan-1964}). Another issue concerns the comparison of the predictions of the DIA with high-Reynolds number experiments (Mou and Weichman \cite{MW-93}, \cite{MW-95}, Eyink \cite{Eyink}).} is currently the only fully self-consistent analytical theory of turbulence in fluids. The formal application of the DIA to a statistical problem is justified only when the nonlinear effects are weak. Since this is not the case in the turbulence problem, there is a need to rationalize the DIA prior to application via a model equation which can be solved exactly. Kraichnan \cite{Kraichnan-1965} applied the DIA to the model equation and modified it by a process, where the non-linear terms are unrestricted, into another equation for which the DIA gives the exact statistical average.

Several mathematical issues associated with the application of the DIA to Kraichnan's \cite{Kraichnan-1961} stochastic model equation were explored by Shivamoggi, et al. \cite{BKS} and Shivamoggi and Tuovila \cite{BT}. Kraichnan \cite{Kraichnan-1961} considered a stochastic model with fluctuations having short-range autocorrelations. Shivamoggi, et al. \cite{BKS} considered instead fluctuations following an \emph{Uhlenbeck-Ornstein} process \cite{UO} with wide-ranging autocorrelations. The stochastic models of Kraichnan \cite{Kraichnan-1961} and Shivamoggi et al. \cite{BKS} are based on a \emph{Markovian} process\footnote{The Markovian property is the statistical analog in random processes of the nearest-neighbor interactions in statistical physics (Lifshits, et al. \cite{LGP}).}. Shivamoggi and Tuovila \cite{BT}, on the other hand, considered a stochastic model based on an \emph{ultra non-Markovian process}\footnote{The non-Markovian property corresponds to long-range interactions in statistical physics, an example of which is Weiss' \cite{Weiss} theory of ferromagnetism.}. All these developments were based on the \emph{Boltzmann-Gibbs} prescription for the underlying entropy measure, which exhibits the \emph{extensivity} property and is suited for \emph{ergodic} systems. However, \emph{non-ergodic}, non-Markovian systems need the introduction of an influence bias discriminating \emph{rare} and \emph{frequent} events explicitly. In order to address this issue, Tsallis \cite{Tsallis-intro} gave a prescription to generalize the concept of entropy exhibiting a \emph{non-extensivity} property to cover such systems.

In this paper, we consider the application of the DIA to a linear damped stochastic oscillator system using a Tsallis type \cite{Tsallis-intro} autocorrelation model with an underlying non-extensive entropy measure, and describe the resulting stochastic process.  In the process, we deduce some apparently novel mathematical properties of the stochastic models associated with the present investigation -- the gamma distribution (Andrews et al. \cite{APS1989}, \cite{AS1990}) and the Tsallis non-extensive entropy \cite{Tsallis-intro}.

\section{Tsallis entropy model}

The Boltzmann-Gibbs entropy of a system with $W$ microstates, each of which has probability $p_i$, is given by the formula
\begin{equation}\label{eq:BG}
    S_{BG} = -k_B\sum_{i=1}^W p_i \ln p_i
\end{equation}
with the normalization condition
\begin{equation}
    \sum_{i=1}^W p_i = 1.
\end{equation}
Here, $k_B$ is the Boltzmann constant, which we will put equal to unity for convenience. In the case of equiprobability, $p_i = 1/W, \forall i$, eq\eqref{eq:BG} then reduces to the Boltzmann ansatz,

\begin{equation} \label{eq:BG-simple}
    S_{BG} = \ln W.
\end{equation}
The Boltzmann-Gibbs entropy is  nonnegative, concave, experimentally robust, and extensive. The latter property implies that for two independent statistical systems $A$ and $B$,

\begin{equation} \label{eq:extensive}
    S_{BG}(A+B) = S_{BG}(A) + S_{BG}(B).
\end{equation}

Just as $S_{BG}$ is postulated, so is its generalization. Tsallis \cite{Tsallis-intro} proposed the following model for entropy (see \ref{appendix1} for a mathematical motivation of this),

\begin{equation}\label{eq:S_q}
    S_q = \frac{1-\sum_{i=1}^W p_i^q}{q-1}
\end{equation}
where $q$ is the parameter characterizing the \emph{non-extensivity} of the entropy. \eqref{eq:S_q} yields, in the Boltzmann-Gibbs (\emph{extensive entropy}, $q\to 1$) limit,

\begin{equation}
    S_1 = \lim_{q\to 1}S_q = -\sum_{i=1}^{W}p_i \ln p_i. \tag{\ref{eq:BG}}
\end{equation}
On maximizing $S_q$, subject to the constraints,
\begin{itemize}
    \item generalized normalization: \begin{equation}
        \sum_i P_i(q) = 1
    \end{equation}
    \item generalized energy conservation: \begin{equation}
        \sum_i P_i(q)E_i = \text{const} \equiv U_q
    \end{equation}
\end{itemize}
where $P_i(q)$ are the \emph{escort} probabilities,

\begin{equation}
    P_i(q) \equiv \frac{p_i^q}{\sum_i p_i^q}
\end{equation}
and $E_i$ is the energy of the $i$th state, we obtain

\begin{equation}\label{eq:max-Sq}
    \overline{p_i}= \frac{1}{Z_q}\biggl[ 1-(1-q)\beta E_i \biggr]^{\frac{1}{1-q}}.
\end{equation}
Here, $\beta$ is the Lagrange multiplier, and $Z_q$ is the partition function ensuring normalization, such that

\begin{equation}
    \sum_i \overline{p_i} = 1.
\end{equation}
In the Boltzmann-Gibbs 
limit, \eqref{eq:max-Sq} gives the usual result,

\begin{equation}\label{eq:BG-limit}
    \lim_{q\to 1} \overline{p_i} \sim e^{-\beta E_i},
\end{equation}
so \eqref{eq:max-Sq} may be viewed as a one-parameter generalization of the Boltzmann-Gibbs formula \eqref{eq:BG-limit}.
On the other hand, for $q < 1$, \eqref{eq:max-Sq} needs the compatibility condition, 

\begin{equation}
    1 - (1-q)\beta E_i < 0 \, : \, \overline{p_i}=0.
\end{equation}

\section{Tsallis autocorrelation model via compound statistics of the Uhlenbeck-Ornstein model for a random process}
Consider a random process described by a real, centered, stationary Gaussian function $b(t)$, which follows the \emph{Uhlenbeck-Ornstein} model \cite{UO}.\footnote{In fact, Doob's theorem \cite{Doob} stipulates that a random process which is stationary, Gaussian, and Markovian follows the Uhlenbeck-Ornstein model.}   We then follow the ansatz of Wilk and Wlodarcyzk \cite{WW} and Beck \cite{Beck} for the Boltzmann distribution, and account for a slowly varying environment of the latter model by stipulating the governing parameter $\lambda$ (the \emph{inverse autocorrelation time}) thereof to follow a \emph{gamma distribution}\footnotemark, hence compounding the statistics underlying the Uhlenbeck-Ornstein model.

\footnotetext{The gamma distribution was previously used in hydrodynamic turbulence to model the kinetic energy dissipation field (Andrews et al. \cite{APS1989}) and scalar-variance dissipation field (Andrews and Shivamoggi \cite{AS1990}).}

Consider the Uhlenbeck-Ornstein model for the \emph{conditional} autocorrelation function of this stationary random process,

\begin{equation}
    r(t', t'', \lambda) \equiv \langle b(t')b(t'') \rangle = \sigma^2 e^{-\lambda(t'-t'')}.
\end{equation}
Suppose the inverse autocorrelation time $\lambda$ is also a random variable distributed according to a gamma distribution,

\begin{equation}\label{eq:lambda^*-density}
    f(\lambda;c) = \frac{1}{a\Gamma(c)}\bigg( \frac{\lambda}{a} \bigg)^{c-1}e^{-\lambda/a}, \,a > 0, \, c > 0.
\end{equation}
\eqref{eq:lambda^*-density} leads to the following expression for the mean $\lambda_0$ and the variance $\sigma^2$ of the random variable $\lambda$,

\begin{subequations}\label{eq:mean-variance}
    \begin{equation}\label{eq:bc}
        \lambda_0 \equiv \langle \lambda \rangle = \int_0^{\infty}\lambda f(\lambda;c)d\lambda = ac
    \end{equation}
    \begin{equation}
        \sigma^2 \equiv \langle \lambda^2 \rangle - \langle \lambda \rangle^2 = a^2c
    \end{equation}
from which, we obtain for the parameters $a$ and $c$,
\begin{equation}\label{eq:ac}
    a = \frac{\sigma^2}{\lambda_0}, \,\, c = \frac{\lambda_0^2}{\sigma^2}.
\end{equation}
\end{subequations}
The \emph{marginal} distribution is then given by averaging the conditional distribution $r(t',t'',\lambda)$ over $\lambda$,

\begin{equation}\label{eq:marginal}
    \begin{split}
        R(t', t'') &= R_q(\tau) \equiv \int_0^{\infty}r(\tau, \lambda)f(\lambda;c)d\lambda \\
        &= \frac{{\sigma}^2}{\Gamma(c)}\int_0^{\infty} \biggl( \frac{\lambda}{a} \biggr)^{c-1} e^{-(1+a\tau)(\lambda/a)}d\left(\frac{\lambda}{a}\right) \\
        &= \frac{{\sigma}^2}{(1+a\tau)^c}, \,\, \tau \equiv t'-t''.
    \end{split}
\end{equation}
Putting (Beck \cite{Beck}),

\begin{equation}\label{eq:c}
    c \equiv \frac{1}{q-1},
\end{equation}
\eqref{eq:marginal} leads to the \emph{Tsallis} \cite{Tsallis-intro} type autocorrelation function,

\begin{equation}\label{eq:B-tsallis}
    R_q(\tau) = {\sigma}^2[1 - (1-q)\lambda_0 \tau]^{\frac{1}{1-q}}
\end{equation}
On writing \eqref{eq:B-tsallis} in terms of the \emph{q-exponential},
\begin{subequations}
\begin{equation}
    R_q(\tau) \equiv \sigma^2 e_q^{-\lambda_0\tau},
\end{equation}
and noting that
\begin{equation}
    \lim_{q\to 1} e_q^{-\lambda_0\tau} = e^{-\lambda_0\tau},
\end{equation}
$q$ is observed to play the role of the \emph{non-extensivity parameter}.
\end{subequations} Furthermore, we have from \eqref{eq:mean-variance} and \eqref{eq:c},

\begin{equation}
    \sigma^2 = (q-1)\lambda_0^2
\end{equation}
which implies,

\begin{equation}
    q \equiv \frac{\langle \lambda^2 \rangle}{\lambda_0^2} > 1.
\end{equation}
On the other hand, we have, in the \emph{zero-dispersion} limit (see \ref{appendix3}),

\begin{equation}\label{eq:gamma-delta}
    \lim_{c\to\infty}f(\lambda;c) = \delta(\lambda-\lambda_0).
\end{equation}
We therefore obtain the Uhlenbeck-Ornstein autocorrelation function in the \emph{extensive entropy} limit, ($q \to 1$),

\begin{equation}
    R_1(\tau) = \lim_{q\to 1} R_q(\tau) \equiv \lim_{q\to 1} \sigma^2e_q^{\lambda_0\tau} = {\sigma}^2 e^{-\lambda_0\tau}
\end{equation}
with the inverse autocorrelation time given by $\lambda_0$.

Some mathematical properties of the marginal distribution \eqref{eq:B-tsallis} are discussed in \ref{appendix2} and \ref{appendix4}.

\section{Linear damped stochastic oscillator}

Consider the linear damped stochastic oscillator described by the initial value problem (IVP) (Zwanzig \cite{Zwanzig}),
\begin{subequations}\label{eq:model}
    \begin{align}
        [\frac{d}{dt} + ib(t)]\hat{G}(t) + \int_0^t \Gamma(t-t')\hat{G}(t') dt' = \delta(t), \, \hat{G}(0)=1
    \end{align}
    where $\Gamma(t)$ is a \textit{history-dependent} damping coefficient,
    \begin{align}
        \Gamma(t) = \nu e^{-\mu t},\,\, \nu \text{ and } \mu > 0.
    \end{align}
\end{subequations}

In the limit $\mu \to \infty$, with $\displaystyle \lim_{\mu \to \infty} \frac{\nu}{\mu} = $ const, the damping process displays zero memory and becomes Markovian, explored first by Kraichnan \cite{Kraichnan-1961} and then by Shivamoggi, et al. \cite{BKS}. In this case, the IVP \eqref{eq:model} leads to
\begin{equation}\label{eq:model-Markov}
    [\frac{d}{dt} + \nu + ib(t)]\hat{G}(t) = \delta(t), \, \hat{G}(0)=1.
\end{equation}
On the other hand, in the limit $\mu \to 0$, the damping process displays infinite memory and becomes ultra non-Markovian. The IVP \eqref{eq:model} then leads to:
\begin{equation}\label{eq:model-nonMarkov}
    [\frac{d}{dt} + ib(t)]\hat{G}(t) + \nu\int_0^t \hat{G}(t') dt' = \delta(t), \, \hat{G}(0)=1
\end{equation}
explored previously by Shivamoggi and Tuovila \cite{BT}.

In both cases, $b(t)$ is a real, centered, stationary Gaussian random function of $t$ described by the Tsallis type \cite{Tsallis-intro} generalization of the Uhlenbeck-Ornstein model \cite{UO},
\begin{equation}\label{eq:autocorr}
    \langle b(t)b(t') \rangle = \sigma^2 e_q^{-\lambda (t-t')} \equiv \sigma^2 [1 - \lambda \ell (t-t')]^{\frac{1}{\ell}}, \quad \ell \equiv 1-q.
\end{equation}
Figure \ref{fig:schematic} shows the schematic of the theoretical considerations that follow below.

\begin{sidewaysfigure}[p]
    \centering
    \includegraphics[width=\maxwidth{\textwidth}]{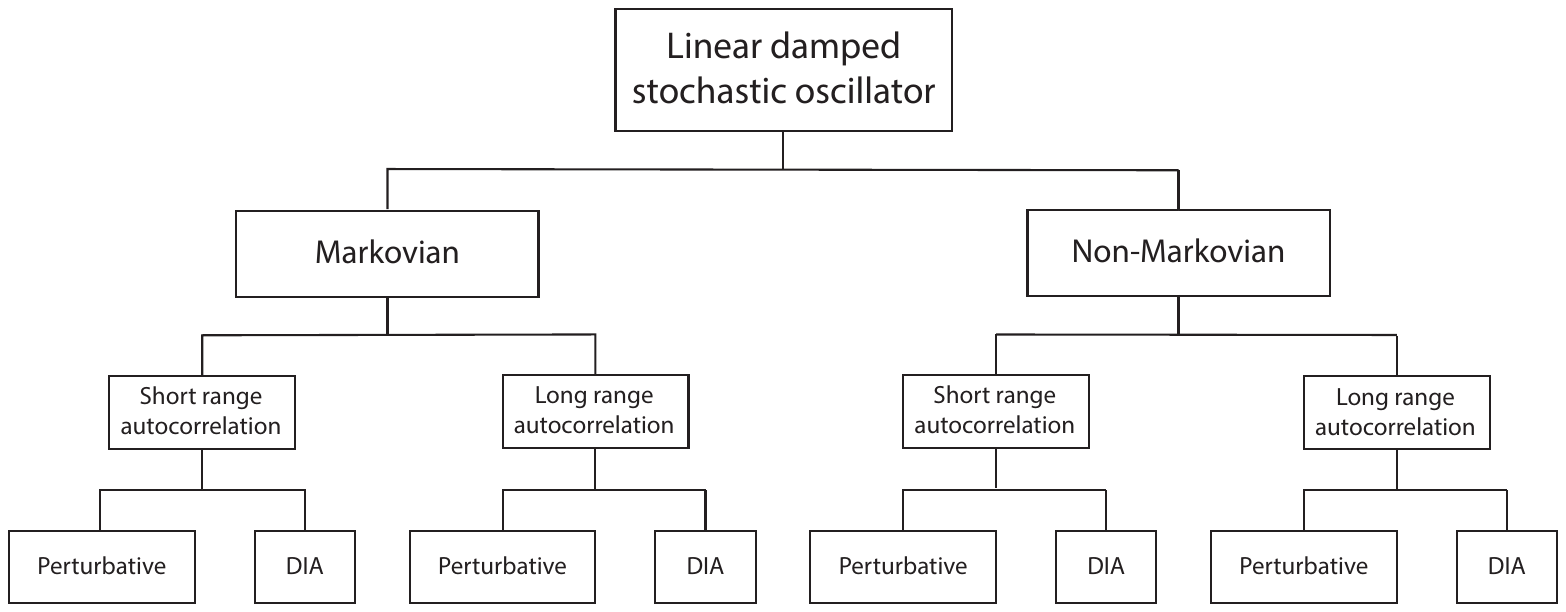}
    \caption{The schematic of the present discussion of the linear damped stochastic oscillator.}\label{fig:schematic}
    \vspace{1in}
\end{sidewaysfigure}

\section{The Markovian Model}

\subsection{The long-range autocorrelation case}\label{sec:long-markov}

We consider first the case where the inverse autocorrelation time $\lambda$ is small.
\subsubsection{Perturbative procedure}\label{sec:black-Markov-perturb}

Applying Keller's perturbation procedure \cite{Keller}, the IVP \eqref{eq:model-Markov} becomes
\begin{equation}\label{eq:perturb-sm}
    \bigg( \frac{d}{dt} + \nu \bigg)G(t) + \int_0^t \langle b(t)b(t') \rangle e^{-\nu(t-t')} G(t') dt' = \delta(t), \qquad G(0) = 1.
\end{equation}
Putting 
\begin{equation}\label{eq:transformtoJ}
    G(t) = e^{-\nu t}J(t)
\end{equation}
and using the Tsallis type \cite{Tsallis-intro} autocorrelation model \eqref{eq:autocorr}, the IVP \eqref{eq:perturb-sm} becomes

\begin{equation}\label{eq:Markov-perturb-J}
    \frac{dJ}{dt} + \sigma^2\int_0^t [1-\lambda \ell(t-t')]^{\frac{1}{\ell}}J(t')dt' = \delta(t), \enspace J(0)=1.
\end{equation}
For small $\lambda$, using the Taylor series approximation \eqref{eq:smalllambdatsallis} (see \ref{appendix2}), we obtain
\begin{equation}\label{eq:mp-sm}
    \frac{dJ}{dt} + \sigma^2\int_0^t [1-\lambda (t-t')]J(t')dt' = \delta(t), \, J(0)=1.
\end{equation}
Applying the Laplace transform, the IVP \eqref{eq:mp-sm} leads to
\begin{equation}\label{eq:mp-sm-L}
    p\JJ(p) + \sigma^2\left[ \frac{1}{p}\left(1 - \frac{\lambda}{p} \right) \right]\JJ(p) = 1.
\end{equation}
Using the \emph{Pad\'e} approximation, eq\eqref{eq:mp-sm-L} yields the functional equation

\begin{equation}\label{eq:mp-sm-L2}
    p\JJ(p) + \sigma^2 \frac{\JJ(p)}{p + \lambda} = 1,
\end{equation}
which agrees with the results obtain by using the Uhlenbeck-Ornstein \cite{UO} model, given in \cite{BKS}.

\subsubsection{DIA Procedure}\label{sec:black-Markov-DIA}

Application of the DIA procedure (Kraichnan \cite{Kraichnan-1961}) to the IVP \eqref{eq:model-Markov} involves replacing the perturbative expression for the deterministic Green's function $e^{-\nu(t-t')}$ in the IVP \eqref{eq:perturb-sm} by the exact expression $G(t-t')$ (which constitutes a renormalization of the perturbative procedure in Section \ref{sec:black-Markov-perturb}), yielding,

\begin{equation}\label{eq:dia-sm}
    \bigg( \frac{d}{dt} + \nu \bigg)G(t) + \int_0^t \langle b(t)b(t')G(t-t') \rangle G(t') dt' = \delta(t), \, G(0)=1.
\end{equation}
Using the \textit{weak statistical-dependence} (also called \emph{quasi-normality}) hypothesis (Kraichnan \cite{Kraichnan-1958}, \cite{Kraichnan-1959}), we may write

\begin{equation}\label{eq:quasi-normal}
    \langle b(t)b(t')G(t-t') \rangle = \langle b(t)b(t') \rangle \langle G(t-t') \rangle.
\end{equation}
Furthermore, using \eqref{eq:autocorr} and \eqref{eq:transformtoJ}, the IVP \eqref{eq:dia-sm} becomes

\begin{equation}\label{eq:markov-dia}
    \frac{dJ}{dt} + \sigma^2\int_0^t [1-\lambda \ell(t-t')]^{\frac{1}{\ell}}J(t-t')J(t')dt' = \delta(t), \enspace J(0)=1.
\end{equation}
For small $\lambda$, using the Taylor series approximation \eqref{eq:smalllambdatsallis} again, we obtain

\begin{equation}\label{eq:mDIA-sm}
    \frac{dJ}{dt} + \sigma^2\int_0^t [1-\lambda (t-t')]J(t-t')J(t')dt' = \delta(t), \, J(0)=1.
\end{equation}
Applying the Laplace transform, the IVP \eqref{eq:mDIA-sm} leads to

\begin{equation}\label{eq:mDIA-sm-L}
    p\JJ(p) + \sigma^2\left[ \JJ(p) + \lambda\JJ'(p) \right]\JJ(p) = 1.
\end{equation}
Noting that we may write, for small $\lambda$, 

\begin{equation}\label{eq:Laplace-taylor}
    \JJ(p) + \lambda \JJ'(p) \approx \JJ(p+\lambda)
\end{equation} 
eq\eqref{eq:mDIA-sm-L} leads to the \emph{functional equation},

\begin{equation}
    p\JJ(p) + \sigma^2 \JJ(p+\lambda)\JJ(p) = 1,
\end{equation}
which agrees again with the results of \cite{BKS} for the Markovian case.

\subsection{The short-range autocorrelation case}

In section \ref{sec:long-markov}, we considered the case for which the inverse autocorrelation time $\lambda$ is small. For large values of $\lambda$ (the \textit{white-noise} limit), the Taylor series approximation \eqref{eq:smalllambdatsallis} is no longer valid. Instead, we use, for small $\ell$, the delta function approximation established in Theorem \ref{thm:large-tsallis} in \ref{appendix2}. 
\subsubsection{Perturbative procedure}\label{sec:white-Markov-perturb}

Applying Keller's perturbation procedure \cite{Keller} and using the Tsallis type \cite{Tsallis-intro} autocorrelation model \eqref{eq:autocorr}, the IVP \eqref{eq:model-Markov} becomes
\begin{equation}
    \frac{dJ}{dt} + \sigma^2\int_0^t [1-\lambda \ell(t-t')]^{\frac{1}{\ell}}J(t')dt' = \delta(t), \enspace J(0)=1. \tag{\ref{eq:Markov-perturb-J}}
\end{equation}
For large $\lambda$, using \eqref{eq:largelambdatsallis}, the IVP \eqref{eq:Markov-perturb-J} becomes

\begin{equation}\label{eq:mp-lg}
    \frac{dJ}{dt} + \sigma^2\int_0^t \frac{1}{\lambda}\delta(t-t') J(t')dt' = \delta(t), \, J(0)=1.
\end{equation}
Applying the Laplace transform, the IVP \eqref{eq:mp-lg} leads to
\begin{equation}\label{eq:mp-lg-L}
    p\JJ(p) + \frac{\sigma^2}{\lambda}\JJ(p) = 1,
\end{equation}
which gives
\begin{equation}\label{eq:m-perturb-short}
    \JJ(p) = \frac{1}{p+ \frac{\sigma^2}{\lambda}}.
\end{equation}
Inverting the Laplace transform, \eqref{eq:m-perturb-short} leads to
\begin{equation}\label{eq:perturb-solution}
    J(t) = e^{-\frac{\sigma^2}{\lambda}t},
\end{equation}
showing a fluctuation-induced exponential attenuation (embodying the \emph{fluctuation-dissipation theorem}, Huang \cite{Huang}).

\subsubsection{DIA procedure}\label{sec:white-Markov-DIA}
Application of the DIA procedure entails again replacing the perturbative expression for the deterministic Green's function in \eqref{eq:Markov-perturb-J} with the exact expression $J(t-t')$, which yields

\begin{equation}
    \frac{dJ}{dt} + \sigma^2\int_0^t [1-\lambda \ell(t-t')]^{\frac{1}{\ell}}J(t-t')J(t')dt' = \delta(t), \enspace J(0)=1.\tag{\ref{eq:markov-dia}}
\end{equation}
For large $\lambda$, using \eqref{eq:largelambdatsallis}, the IVP \eqref{eq:markov-dia} becomes

\begin{equation}\label{eq:m-d-lg}
    \frac{dJ}{dt} + \sigma^2\int_0^t \frac{1}{\lambda}\delta(t-t') J(t-t')J(t')dt' = \delta(t), \,J(0)=1,
\end{equation}
which leads to 

\begin{equation}\label{eq:m-d-simplified}
    \frac{dJ}{dt} + \frac{\sigma^2}{\lambda}J(0)J(t) = \delta(t), \, J(0)=1
\end{equation}
and the solution is
\begin{equation}
    J(t) \sim e^{-\frac{\sigma^2}{\lambda}t},
\end{equation}
which is exactly the perturbative result \eqref{eq:perturb-solution}. This indicates that in the \textit{white-noise} limit the non-perturbative aspects excluded by Keller's perturbative procedure \cite{Keller} are minimized.

\subsection{The large time limit}

For large $t$, we need to consider $\lambda$ to be small to ensure that the q-exponential given by \eqref{eq:autocorr} is nonzero. 

Noting that large $t$ corresponds to small $p$,
eq\eqref{eq:mp-sm-L2} leads to

\begin{equation}
    p\JJ(p) + \frac{\sigma^2}{\lambda}\JJ(p) \approx 1
\end{equation}
from which, we obtain

\begin{equation}\label{eq:lgt-markov}
    \JJ(p) \sim \frac{1}{p+\frac{\sigma^2}{\lambda}}.
\end{equation}
Upon inverting the Laplace transform, \eqref{eq:lgt-markov} leads to

\begin{equation}
    J(t) \sim e^{-\frac{\sigma^2}{\lambda}t},
\end{equation}
which agrees with the large time limit of the solution for the Uhlenbeck-Ornstein model (Shivamoggi et al. \cite{BKS}).

\section{The Non-Markovian Model}
\subsection{The long-range autocorrelation case}\label{sec:long-range-nonMarkov}
\subsubsection{Perturbative procedure}\label{sec:black-nonMarkov-perturb}

For the non-Markovian case, on applying Keller's perturbative procedure \cite{Keller}, the IVP \eqref{eq:model-nonMarkov} becomes 

\begin{equation}\label{eq:non-markov-perturb}
    \frac{dG}{dt} + \nu\int_0^t G(t')dt' + \sigma^2\int_0^t \langle b(t)b(t') \rangle \cos\sqrt{\nu}(t-t') G(t')dt' = \delta(t), \, G(0)=1.
\end{equation}
Using the Tsallis type \cite{Tsallis-intro} autocorrelation model \eqref{eq:autocorr}, the IVP \eqref{eq:non-markov-perturb} becomes

\begin{equation}\label{eq:nm-p}
    \frac{dG}{dt} + \nu\int_0^t G(t')dt' + \sigma^2\int_0^t  [1-\lambda \ell(t-t')]^{\frac{1}{\ell}} \cos\sqrt{\nu}(t-t') G(t')dt' = \delta(t), \, G(0)=1.
\end{equation}
For small $\lambda$, using the Taylor series approximation \eqref{eq:smalllambdatsallis}, the IVP \eqref{eq:nm-p} becomes

\begin{equation}\label{eq:nm-p-sm}
    \frac{dG}{dt} + \nu\int_0^t G(t')dt' + \sigma^2\int_0^t  [1-\lambda (t-t')] \cos\sqrt{\nu}(t-t') G(t')dt' = \delta(t), \,G(0)=1.
\end{equation}
Applying the Laplace transform, the IVP \eqref{eq:nm-p-sm} leads to

\begin{equation}\label{eq:nm-p-sm-L}
    p\,\GG(p) + \nu\frac{\GG(p)}{p} + \sigma^2 \left[ \frac{p}{p^2 + \nu} + \lambda \frac{d}{dp}\bigg( \frac{p}{p^2 + \nu}\bigg) \right]\GG(p) = 1.
\end{equation}
Noting that we may write, for small $\lambda$,
\begin{equation} \label{eq:taylor}
    \frac{p}{p^2 + \nu} + \lambda \frac{d}{dp}\bigg( \frac{p}{p^2 + \nu}\bigg) \approx \frac{p+\lambda}{(p+\lambda)^2 + \nu},
\end{equation}
eq\eqref{eq:nm-p-sm-L} becomes

\begin{equation}
    \GG(p) \left[ p + \frac{\nu}{p} + \sigma^2 \frac{p+\lambda}{(p+\lambda)^2+\nu} \right] = 1,
\end{equation}
which agrees with the results given in \cite{BT} for the non-Markovian case.

\subsubsection{DIA procedure}\label{sec:black-nonMarkov-DIA}

Application of the DIA procedure (Kraichnan \cite{Kraichnan-1961}) entails again replacing the perturbative expression for the deterministic Green's function in the IVP \eqref{eq:non-markov-perturb} with the exact expression $G(t-t')$, which yields

\begin{equation}\label{eq:non-markov-DIA}
    \frac{dG}{dt} + \nu\int_0^t G(t')dt' + \sigma^2\int_0^t \langle b(t)b(t') G(t-t') \rangle G(t')dt' = \delta(t), \, G(0)=1.
\end{equation}
Using the quasi-normality hypothesis \eqref{eq:quasi-normal}, the Tsallis type \cite{Tsallis-intro} autocorrelation model \eqref{eq:autocorr}, and the small-$\lambda$ approximation \eqref{eq:smalllambdatsallis}, the IVP \eqref{eq:non-markov-DIA} leads to 

\begin{equation}\label{eq:nm-d-sm}
    \frac{dG}{dt} + \nu\int_0^t G(t')dt' + \sigma^2\int_0^t  [1-\lambda (t-t')] G(t-t') G(t')dt' = \delta(t), G(0)=1.
\end{equation}
Applying the Laplace transform, the IVP \eqref{eq:nm-d-sm} leads to

\begin{equation}\label{eq:nm-d-sm-L}
    p\,\GG(p) + \nu\frac{\GG(p)}{p} + \sigma^2 \left[ \GG(p) - \lambda\bigg(-\frac{d}{dp}\bigg)\GG(p) \right]\GG(p) = 1.
\end{equation}
Eq\eqref{eq:nm-d-sm-L} may be rearranged, for small $\lambda$, to yield the \emph{continued fraction} solution
\begin{equation}
    \GG(p) = \frac{1}{p + \frac{\nu}{p} + \sigma^2 \GG(p+\lambda)},
\end{equation}
which agrees again with the results given in \cite{BT} for the non-Markovian case.

\subsection{The short-range autocorrelation case}

In Section \ref{sec:long-range-nonMarkov}, we considered the case for which the inverse autocorrelation time $\lambda$ is small.  For large values of $\lambda$ (the white-noise limit), the Taylor series approximation \eqref{eq:smalllambdatsallis} is no longer valid.  Instead, we again use, for small $\ell$, the delta function approximation established in Theorem \ref{thm:large-tsallis} in \ref{appendix2}.

\subsubsection{Perturbative procedure}\label{sec:white-nonMarkov-perturb}
Applying Keller's perturbation procedure \cite{Keller} and using the Tsallis type \cite{Tsallis-intro} autocorrelation model \eqref{eq:autocorr}, the IVP \eqref{eq:model-nonMarkov} leads to

\begin{equation}\label{eq:white-nM-perturb-model}
    \frac{dG}{dt} + \nu\int_0^t G(t')dt' + \sigma^2\int_0^t  [1-\lambda \ell(t-t')]^{\frac{1}{\ell}} \cos\sqrt{\nu}(t-t') G(t')dt' = \delta(t),\, G(0)=1. \tag{\ref{eq:nm-p}}
\end{equation}
For large $\lambda$, using \eqref{eq:largelambdatsallis}, the IVP \eqref{eq:nm-p} becomes

\begin{equation}\label{eq:nm-p-lg}
    \frac{dG}{dt} + \nu\int_0^t G(t')dt' + \sigma^2\int_0^t  \frac{1}{\lambda}\delta(t-t') \cos\sqrt{\nu}(t-t') G(t')dt' = \delta(t), \, G(0)=1.
\end{equation}
Applying the Laplace transform, the IVP \eqref{eq:nm-p-lg} leads to
\begin{equation}\label{eq:nm-p-lg-L}
    p\,\GG(p) + \nu\frac{\GG(p)}{p} + \frac{\sigma^2}{\lambda}\GG(p) = 1
\end{equation}
which gives,
\begin{equation}\label{eq:nm-p-lg-L2}
    \GG(p) = \frac{p}{p^2 + \frac{\sigma^2}{\lambda}p + \nu}
\end{equation}
Upon inverting the Laplace transform, \eqref{eq:nm-p-lg-L2} leads to
\begin{equation}\label{eq:nM-p-solution}
    G(t) = e^{(-\sigma^2/2\lambda)t}\cos\sqrt{\nu-\frac{\sigma^4}{4\lambda^2}}\,t - e^{(-\sigma^2/2\lambda)t}\sin\sqrt{\nu-\frac{\sigma^4}{4\lambda^2}}\,t,
\end{equation}
showing again a fluctuation-induced exponential attenuation.

\subsubsection{DIA procedure}\label{sec:white-nonMarkov-DIA}

Applying the DIA procedure (Kraichnan \cite{Kraichnan-1961}) again as before, the IVP \eqref{eq:model-nonMarkov} leads, in place of \eqref{eq:white-nM-perturb-model}, to
\begin{equation}
    \frac{dG}{dt} + \nu\int_0^t G(t')dt' + \sigma^2\int_0^t  [1-\lambda \ell(t-t')]^{\frac{1}{\ell}} G(t-t') G(t')dt' = \delta(t), \, G(0)=1. \tag{\ref{eq:non-markov-DIA}}
\end{equation}
For large $\lambda$, using \eqref{eq:largelambdatsallis}, the IVP \eqref{eq:non-markov-DIA} becomes

\begin{equation}\label{eq:nm-d-lg}
    \frac{dG}{dt} + \nu\int_0^t G(t')dt' + \sigma^2\int_0^t  \frac{1}{\lambda}\delta(t-t') G(t-t') G(t')dt' = \delta(t), \, G(0)=1.
\end{equation}
Upon applying the Laplace transform, the IVP \eqref{eq:nm-d-lg} leads to
\begin{equation}\label{eq:nm-d-lg-L}
    p\,\GG(p) + \nu\frac{\GG(p)}{p} + \frac{\sigma^2}{\lambda}\GG(p) = 1.
\end{equation}
Eq\eqref{eq:nm-d-lg-L} gives
\begin{equation}
    \GG(p) = \frac{p}{p^2 + \frac{\sigma^2}{\lambda}p + \nu},
\end{equation}
which is the same as the perturbative result \eqref{eq:nM-p-solution}, indicating again that the \textit{white-noise} limit minimizes the non-perturbative aspects excluded by Keller's perturbative procedure \cite{Keller}, as in the Markovian case.

\subsection{The large-time limit}

For large $t$, the Green's function $\cos\sqrt{\nu}(t-t')$ in the non-Markovian IVP \eqref{eq:nm-p} oscillates rapidly, making a net-zero contribution. The IVP \eqref{eq:nm-p} then reduces to

\begin{equation}\label{eq:nm-lgtime}
    \frac{dG}{dt} + \nu\int_0^t G(t')dt' = \delta(t), \, G(0)=1.
\end{equation}
Applying the Laplace transform, \eqref{eq:nm-lgtime} yields

\begin{equation}
    p\,\GG(p) + \frac{\nu}{p}\GG(p) = 1,
\end{equation}
from which,

\begin{equation}\label{eq:nm-largetime-laplace}
    \GG(p) = \frac{p}{p^2 + \nu}.
\end{equation}
Inverting the Laplace transform, \eqref{eq:nm-largetime-laplace} leads to

\begin{equation}
    G(t) = \cos(\sqrt{\nu}t),
\end{equation}
which shows an indefinite oscillation in the limit $t \to \infty$, befitting a random process with infinite memory.

\section{Discussion}

The direct interaction approximation (DIA) developed by Kraichnan \cite{Kraichnan-1961}, \cite{Kraichnan-1965} continues to provide the only fully self-consistent analytical theory of turbulence in fluids.  The requirement of weak nonlinear effects for formal application of the DIA to a statistical problem necessitates a rationalization of the DIA in dealing with a counterexample like the turbulence problem.  Kraichnan \cite{Kraichnan-1965} sought to accomplish this by applying the DIA to a model equation which can be solved exactly, and several mathematical issues associated with this problem were explored by Shivamoggi et al. \cite{BKS} and Shivamoggi and Tuovila \cite{BT}.  These developments (\cite{Kraichnan-1961,Kraichnan-1965, BKS, BT}) were based on the Boltzmann-Gibbs prescription for the underlying entropy measure, which exhibits the extensivity property.  In this paper, we have considered the application of the DIA to non-ergodic stochastic systems using a Tsallis type \cite{Tsallis-intro} autocorrelation model.  As an example, we have considered a linear damped stochastic oscillator system, and have analyzed the Markovian and non-Markovian cases separately.  The DIA solutions have been compared with those given by Keller's perturbative procedure \cite{Keller}.  The non-perturbative aspects excluded by the latter procedure are found to be minimized in the white-noise limit.  On the other hand, in the opposite limit, the physical dissimilarities between Tsallis type and Uhlenbeck-Ornstein models don't seem to materialize, and the two models are found to yield the same result.  In the process of these developments, we have also deduced some apparently novel mathematical properties of the stochastic models involved - the gamma distribution (Andrews et al. \cite{APS1989}, Andrews and Shivamoggi \cite{AS1990}) and the Tsallis non-extensive entropy (Tsallis \cite{Tsallis-intro}).

\section*{Acknowledgements}
We acknowledge the benefit of helpful discussions with Professor Larry Andrews, and are thankful to Professor Katepalli Sreenivasan for his helpful remarks.

\newcounter{oldsection}
\setcounter{oldsection}{\thesection}

\appendix
\renewcommand{\thesection}{Appendix \Alph{section}}

\setcounter{equation}{0}
\renewcommand{\theequation}{A.\arabic{equation}}
\section[Appendix]{Mathematical motivation for Tsallis entropy}\label{appendix1}

One may motivate the generalization of the Boltzmann-Gibbs entropy $S_{BG}$ to the Tsallis entropy by introducing (Tsallis \cite{TsallisReview}) the following initial-value problem (IVP),
\begin{equation}\label{eq:odey}
    \frac{dy}{dx} = y,\enspace y(0) = 1,
\end{equation}
which has the solution is $y=e^x$. Its inverse function is $y = \ln(x),$ which has the same functional form as the Boltzmann-Gibbs entropy \eqref{eq:BG-simple}, and satisfies the additive property
\begin{equation}
    \ln(x_Ax_B) = \ln(x_A) + \ln(x_B).
\end{equation}

Consider next a one-parameter IVP,
\begin{equation}\label{eq:odeyq}
    \frac{dy}{dx} = y^q,\enspace y(0) = 1,
\end{equation}
which for $q=1$, leads to the IVP \eqref{eq:odey}. This generalization has the advantage of having only one parameter, but at the expense of the loss of linearity. The solution of the IVP \eqref{eq:odeyq} is the \emph{q-exponential} function
\begin{equation}\label{eq:qexp}
    y = [1+(1-q)x]^{\frac{1}{1-q}} \equiv e_q^x \qquad (e_1^x = e^x),
\end{equation}
whose inverse is the \emph{q-logarithmic} function
\begin{equation}\label{eq:lnq}
    y = \frac{x^{1-q} - 1}{1-q} \equiv \ln_q(x) \qquad (\ln_1(x) = \ln(x)).
\end{equation}
This function satisfies the pseudo-additive property
\begin{equation}\label{eq:psudeo-additive}
    \ln_q(y_Ay_B) = \ln_q(y_A) + \ln_q(y_B) + (1-q)\ln_q(y_A)\ln_q(y_B).
\end{equation}

The Boltzmann-Gibbs entropy \eqref{eq:BG} may be rewritten as
\begin{equation}
    S_{BG} = \left\langle \ln\left(\frac{1}{p_i}\right) \right\rangle.
\end{equation}
The quantity $\ln(\frac{1}{p_i})$ is called \textit{``surprise''}. The Tsallis entropy may therefore be defined by introducing the \textit{``q-surprise''} $\ln_q(1/p_i)$, and using \eqref{eq:lnq}, we obtain
\begin{equation}\label{eq:tsallis}
    S_q \equiv \left\langle\ln_q\left(\frac{1}{p_i}\right) \right\rangle = \sum_{i=1}^W p_i\ln_q\left(\frac{1}{p_i}\right) = \frac{1- \sum_{i=1}^Wp_i^q}{q-1}.
\end{equation}

\setcounter{equation}{0}
\renewcommand{\theequation}{B.\arabic{equation}}

\section{Zero dispersion limit of the Gamma distribution}\label{appendix3}

\begin{lemma}
    In the \emph{zero dispersion} limit $c \to \infty$, the gamma distribution reduces to $\delta(\lambda  - \lambda_0)$,
    \begin{equation}
        \lim_{c\to\infty} f(\lambda;c) \equiv \lim_{c\to\infty} \frac{1}{a\Gamma(c)}\biggl(\frac{\lambda }{a}\biggr)^{c-1}e^{-\lambda /a} = \delta(\lambda  - \lambda_0)\tag{\ref{eq:lambda^*-density}}, \, \lambda_0=ac.
    \end{equation}
\end{lemma}

\begin{proof}
Using \eqref{eq:bc}, \eqref{eq:lambda^*-density} becomes

\begin{equation}\label{eq:all-c}
    \lim_{c\to\infty}f(\lambda; c) = \lim_{c\to\infty} \frac{c}{\lambda_0\Gamma(c)}\biggl(\frac{\lambda c}{\lambda_0}\biggr)^{c-1}e^{-\lambda c/\lambda_0}.
\end{equation}
Using the following asymptotic result for the gamma function (Andrews \cite{Andrews}),

\begin{equation}\label{eq:Gamma-bounds}
    \Gamma(x) \sim \sqrt{2\pi}x^xe^{-x},
\end{equation}
we have, for $\lambda  \neq \lambda_0$,

\begin{equation}
    \begin{split}\label{eq:lim-is-zero}
        \lim_{c\to\infty} \frac{c^c}{\lambda_0\Gamma(c)}\biggl(\frac{\lambda }{\lambda_0}\biggr)^{c-1}e^{-\lambda c/\lambda_0} &\sim \lim_{c\to\infty} \frac{c^c}{\lambda_0\sqrt{2\pi}c^{c}e^{-c}}\biggl(\frac{\lambda }{\lambda_0}\biggr)^{c}e^{-\lambda c/\lambda_0} \\ 
        &\sim \lim_{c\to\infty} \frac{1}{\lambda_0\sqrt{2\pi}}\exp{\biggl[1-\frac{\lambda }{\lambda_0}+\ln(\frac{\lambda }{\lambda_0})\biggr]c} \\
        &\Rightarrow 0,
        \end{split}
\end{equation}
on noting $[1- (\lambda/\lambda_0)+\ln(\lambda /\lambda_0)] < 0$ (see Figure \ref{fig:neg}).

\begin{figure}[p]
    \vspace{1in}
    \centering
    \includegraphics[width=\maxwidth{0.9\textwidth}]{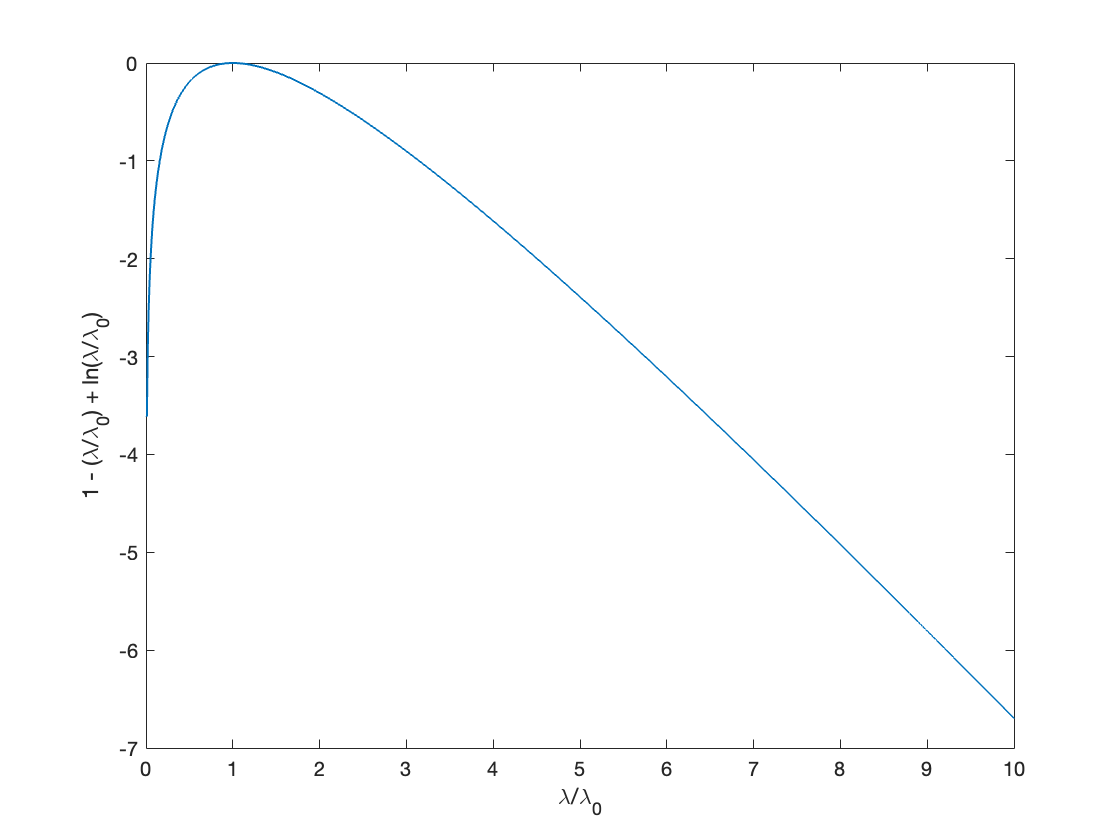}
    \caption{Plot showing $1- (\lambda/\lambda_0)+\ln(\lambda /\lambda_0) < 0$.}\label{fig:neg}
    \vspace{1in}
\end{figure}
Furthermore,

\begin{equation}\label{eq:norm-one}
    \int_0^{\infty} f(\lambda; c)d\lambda = \frac{c^c}{\lambda_0^c\Gamma(c)}\int_0^{\infty}\lambda^{c-1}e^{-\lambda c/\lambda_0} = \frac{c^c}{\lambda_0^c\Gamma(c)} \frac{\Gamma(c)}{(c/\lambda_0)^c} = 1, \,\, \forall c.
\end{equation}

Finally, if $g(\lambda)$ is a smoothly varying function of $\lambda$, we have
\begin{align}
    \int_0^{\infty} g(\lambda)\lim_{c\to\infty}f(\lambda;c)d\lambda &= \int_0^{\lambda_0-\varepsilon} g(\lambda)\lim_{c\to\infty}f(\lambda;c)d\lambda &\hspace{-1em}+ \int_{\lambda_0-\varepsilon}^{\lambda_0 \notag+\varepsilon} g(\lambda)\lim_{c\to\infty}f(\lambda;c)d\lambda \\ &&+ \int_{\lambda_0+\varepsilon}^{\infty} g(\lambda)\lim_{c\to\infty}f(\lambda;c)d\lambda \\
    &\approx g(\lambda_0)\int_{\lambda_0-\varepsilon}^{\lambda_0+\varepsilon}\lim_{c\to\infty}f(\lambda;c)d\lambda \label{eq:ints-are-zero} \\
    &\approx g(\lambda_0) \int_0^{\infty} f(\lambda;c)d\lambda \\
    &\approx g(\lambda_0) \label{eq:sift},
\end{align}
where \eqref{eq:ints-are-zero} follows from \eqref{eq:lim-is-zero}, and \eqref{eq:sift} follows from \eqref{eq:norm-one}. 

Therefore, \[\lim_{c\to\infty}f(\lambda; c) = \delta(\lambda - \lambda_0).\]  This result is also numerically verified (see Figure \ref{fig:gammadelta}).
\end{proof}

\begin{figure}[p]
    \centering
    \includegraphics[width=\maxwidth{0.9\textwidth}]{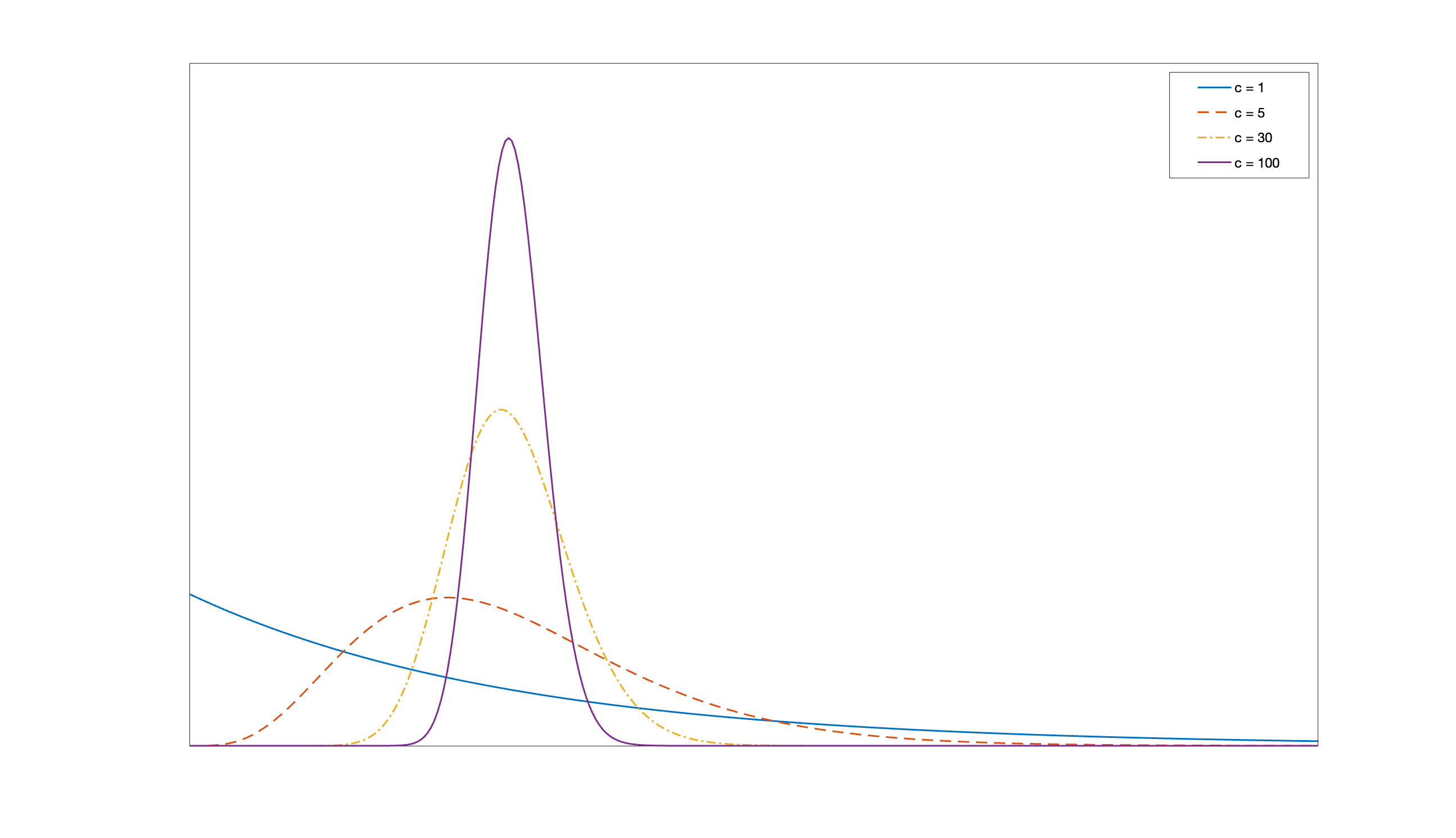}
    \caption{The Gamma distribution as a function of the dispersion parameter $c$.}\label{fig:gammadelta}
    \vspace{1in}
\end{figure}

\setcounter{equation}{0}
\renewcommand{\theequation}{C.\arabic{equation}}

\section{Mathematical properties of Tsallis type autocorrelation function}\label{appendix2}

The Tsallis type generalization of the autocorrelation may be motivated via the solution of the IVP associated with the autocorrelation function, 
\begin{equation}\label{eq:ode} 
    \begin{rcases}
        \dfrac{d\rho}{dt} = -\lambda\rho^q, \, q \neq 1 \\ 
        \rho(0)=1
    \end{rcases}.
\end{equation}
The Uhlenbeck-Ornstein model, used in Shivamoggi et al. \cite{BKS} and Shivamoggi and Tuovila \cite{BT}, is a special case of \eqref{eq:ode}, corresponding to $q=1$.

The solution to the IVP \eqref{eq:ode} is given by
\begin{equation}\label{eq:q-exp}
    y = [1 - \lambda(1-q)t]^{\frac{1}{1-q}} \equiv e_q^{-\lambda t}, \qquad (e_1^{-\lambda t} = e^{-\lambda t})
\end{equation}
where $\lambda$ may be interpreted as the inverse autocorrelation time.
For notational simplicity, we may write
\begin{equation}\label{eq:ql-exp}
    e_q^{-\lambda t} \equiv [1 - \lambda \ell t]^{\frac{1}{\ell}}, \qquad \ell \equiv 1-q.
\end{equation}
Here we discuss the asymptotic properties of \eqref{eq:q-exp}.

\begin{lemma}\label{lemma:limits} For all $t \geq 0$, the iterated limits
    \begin{equation}
        \lim_{\lambda \to \infty} \lim_{\ell \to 0} (1-\lambda \ell t)^{\frac{1}{\ell}} = \lim_{\ell \to 0} \lim_{\lambda \to \infty} (1-\lambda \ell t)^{\frac{1}{\ell}} = 0 
    \end{equation}
    exist and commute, but the simultaneous limit
    \begin{equation*}
        \lim_{\substack{\lambda \to \infty \\ \ell\to 0}} (1-\lambda \ell t)^{\frac{1}{\ell}}
    \end{equation*}
    does not exist.
\end{lemma}

\begin{proof}
    We begin by evaluating the limit, $\ell \to 0$, first. We note that 
    \[{\lim_{\ell \to 0}} (1-\lambda \ell t)^{\frac{1}{\ell}} = e^{-\lambda t}.\] 
    Therefore, for all $t \neq 0$,
    \begin{equation*}
        \lim_{\lambda \to \infty} \lim_{\ell \to 0} (1-\lambda \ell t)^{\frac{1}{\ell}} = \lim_{\lambda \to \infty} e^{-\lambda t} = 0.
    \end{equation*}
    
    On the other hand, when evaluating the limit $\lambda \to \infty$ first, we consider separately the left and right limits $\ell \to 0^-$ and $\ell \to 0^+$. When approaching $\ell \to 0^+$, for sufficiently large $\lambda$, we have $(1 - \lambda\ell t) < 0.$  Noting, $e_q^{-\lambda t} = \max\{0,(1-\lambda \ell t)^{\frac{1}{\ell}}\}$, we have
    \begin{equation*}
        \lim_{\ell \to 0^+} \lim_{\lambda \to \infty} (1-\lambda \ell t)^{\frac{1}{\ell}} = 0.
    \end{equation*}
    When $\ell \to 0^-$, we have $\frac{1}{\ell} < 0$ and hence
    \begin{equation*}
        \lim_{\ell \to 0^-} \lim_{\lambda \to \infty} \bigg(\frac{1}{1+\lambda \abs{\ell} t}\bigg)^{\abs{\frac{1}{\ell}}} = 0.
    \end{equation*}
    Therefore,

    \begin{equation*}
        \lim_{\ell \to 0} \lim_{\lambda \to \infty} (1-\lambda \ell t)^{\frac{1}{\ell}} = 0
    \end{equation*}
    for all $t \neq 0$.

    For the simultaneous limit, let $N$ be a natural number. First, consider $\ell = \frac{1}{N^2}$ and $\lambda = N$. Then
    \begin{equation*}
        \lim_{\substack{\lambda \to \infty \\ \ell\to 0}} (1-\lambda \ell t)^{\frac{1}{\ell}} = \lim_{N \to \infty} \bigg(1 - \frac{t}{N}\bigg)^N = e^t.
    \end{equation*}
    On the other hand, if we consider $\ell = \frac{1}{N^3}$ and $\lambda = N$, we obtain
    \begin{equation*}
        \lim_{\substack{\lambda \to \infty \\ \ell\to 0}} (1-\lambda \ell t)^{\frac{1}{\ell}} = \lim_{N \to \infty} \bigg(1 - \frac{t}{N^2}\bigg)^N = 1.
    \end{equation*}
    Therefore, the simultaneous limit $\displaystyle \lim_{\substack{\lambda \to \infty \\ \ell\to 0}} (1-\lambda \ell t)^{\frac{1}{\ell}}$ does not exist, even though both iterated limits exist and are equal.

\end{proof}

\begin{theorem}\label{thm:large-tsallis}
    For small values of $\lambda$,
    \begin{equation}\label{eq:smalllambdatsallis}
        [1-\lambda \ell t]^{\frac{1}{\ell}} \sim 1 - \lambda t.
    \end{equation}
    For very large values of $\lambda$,
        \begin{equation}\label{eq:largelambdatsallis}
            \lim_{\lambda \to \infty} \lim_{\ell\to 0} (1-\lambda \ell t)^{\frac{1}{\ell}} \sim \frac{1}{\lambda}\delta(t).
        \end{equation}
    \end{theorem}

\begin{proof}
    For small values of $\lambda$, on using a first order Taylor series approximation \cite{TsallisReview}, we have

    \begin{equation}
        [1-\lambda \ell t]^{\frac{1}{\ell}} \sim 1 - \lambda t. \tag{\ref{eq:smalllambdatsallis}}
    \end{equation}
    For large values of $\lambda$, we have, from Lemma \ref{lemma:limits} that, for all $t \neq 0$, 
    \[\lim_{\lambda \to \infty} \lim_{\ell\to 0} (1-\lambda \ell t)^{\frac{1}{\ell}} = 0.\]
    Furthermore, we have
    \begin{subequations}
        \begin{align}
            \int_0^\infty \lim_{\lambda \to \infty} \lim_{\ell \to 0} \lambda(1-\lambda \ell t)^{\frac{1}{\ell}} dt &= \int_0^\infty \lim_{\lambda \to \infty} \lambda e^{-\lambda t} dt \\
            & = \lim_{\lambda \to \infty} \int_0^\infty \lambda e^{-\lambda t} dt \label{eq:DCT} \\
            & = \lim_{\lambda \to \infty} 1 = 1.
        \end{align}
    \end{subequations}

\end{proof}

\setcounter{equation}{0}
\renewcommand{\theequation}{D.\arabic{equation}}

\section{Double integral of Tsallis type autocorrelation}\label{appendix4}
Consider the double integral of the autocorrelation,
\begin{equation}\label{eq:auto}
    I_q(t) \equiv \frac{1}{2}\int_0^t\int_0^t \langle b(t')b(t'') \rangle dt' dt''.
\end{equation}
For a stationary process, the autocorrelation becomes
\begin{subequations}
    \begin{equation}
        \langle b(t')b(t'') \rangle \equiv g(t'-t'')
    \end{equation}
and \eqref{eq:auto} may then be transformed via a change of variables into the single integral

\begin{equation}\label{eq:autocorr-int}
   I_q(t) = \int_0^t (t-\tau)g(\tau) d\tau.
\end{equation}
\end{subequations}
Using the Tsallis type \cite{Tsallis-intro} prescription, we have

\begin{equation}
    g(\tau) = e_q^{-\lambda \tau} = [1 - \lambda\ell\tau]^{\frac{1}{\ell}}.
\end{equation} 
\eqref{eq:auto} then becomes

\begin{equation}
    I_q(t) = \int_0^t (t-\tau)[1 - \ell \lambda\tau]^{\frac{1}{\ell}} d\tau,
\end{equation}
which can be evaluated exactly,

\begin{equation}
    I_q(t) = \frac{t}{\lambda(1+\ell)} + \frac{1}{\lambda^2(1+\ell)(1+2\ell)}[1-\lambda \ell t]^{\frac{1}{\ell}+2} - \frac{1}{\lambda^2(1+\ell)(1+2\ell)}
\end{equation}
for $\ell \neq 0$ $(q \neq 1)$.

Numerical evaluations demonstrate several properties of $I_q(t)$:
\begin{property}
    For $q < 1$, $I_q(t)$ increases as $q$ increases (see Figure \ref{fig:vary-q}).
\end{property}
\begin{property}
    For $q > 1$, $I_q(t)$ decreases as $\lambda$ increases (see Figure \ref{fig:vary-lambda}).
\end{property}

\begin{figure}[p]
    \centering
    \includegraphics[width=\maxwidth{0.9\textwidth}]{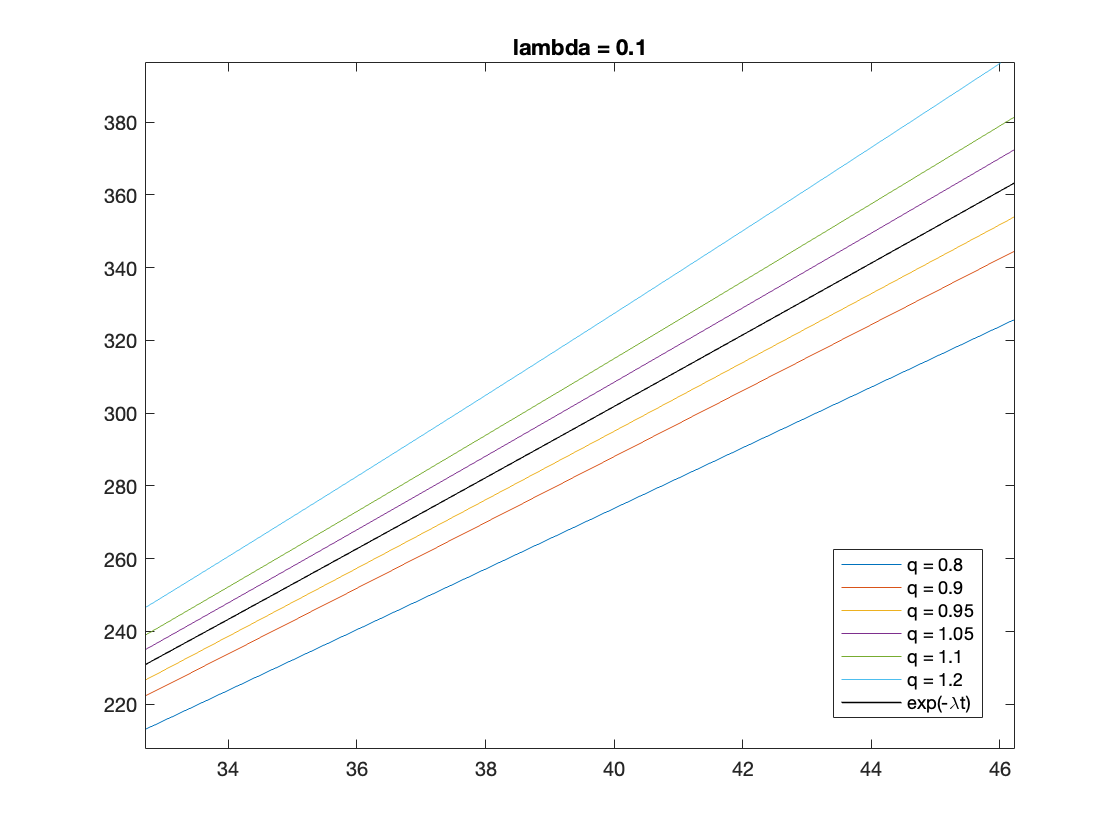}
    \caption{Plot showing $I_q(t)$ vs $q$. Here the parameter $\lambda =0.1$, but the behavior is unchanged when varying $\lambda$.}\label{fig:vary-q}
\end{figure}

\begin{figure}[p]
    \centering
    \begin{subfigure}{6.5in}
        \centering
        \includegraphics[width=6.5in,height=4.2in,keepaspectratio,trim={1in 2.5in 1in 2.5in}]{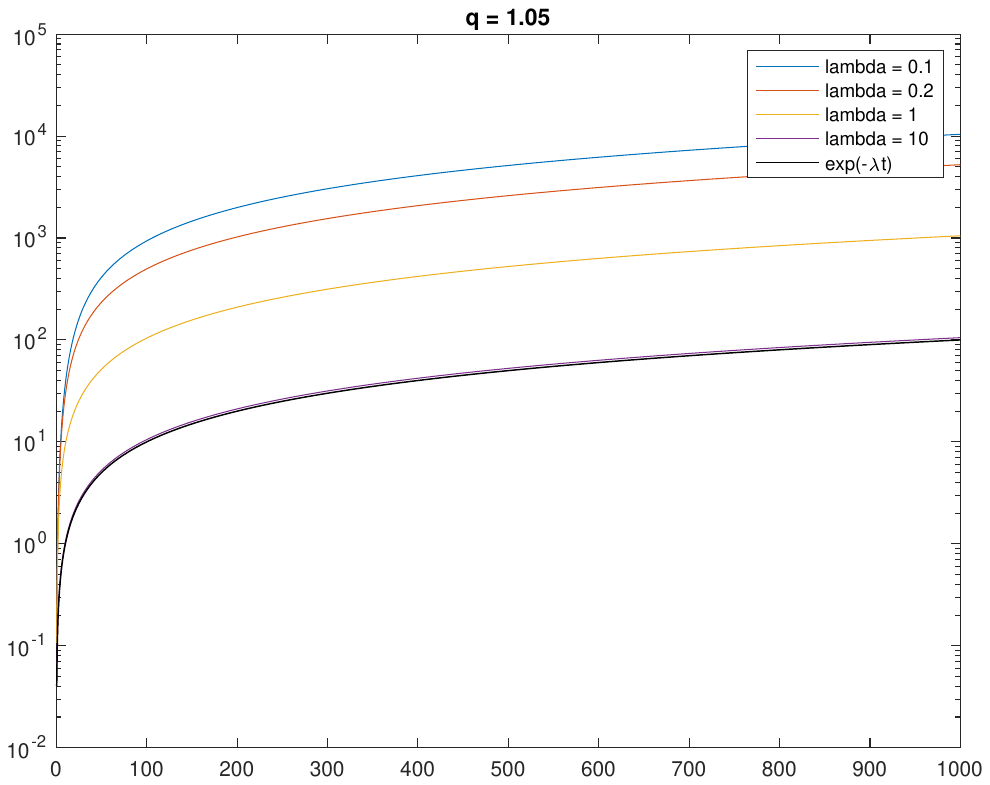}
    \end{subfigure}
    
    \begin{subfigure}{6.5in}
        \centering
        \includegraphics[width=6.5in,height=4.2in,keepaspectratio,trim={1in 2.5in 1in 2.5in}]{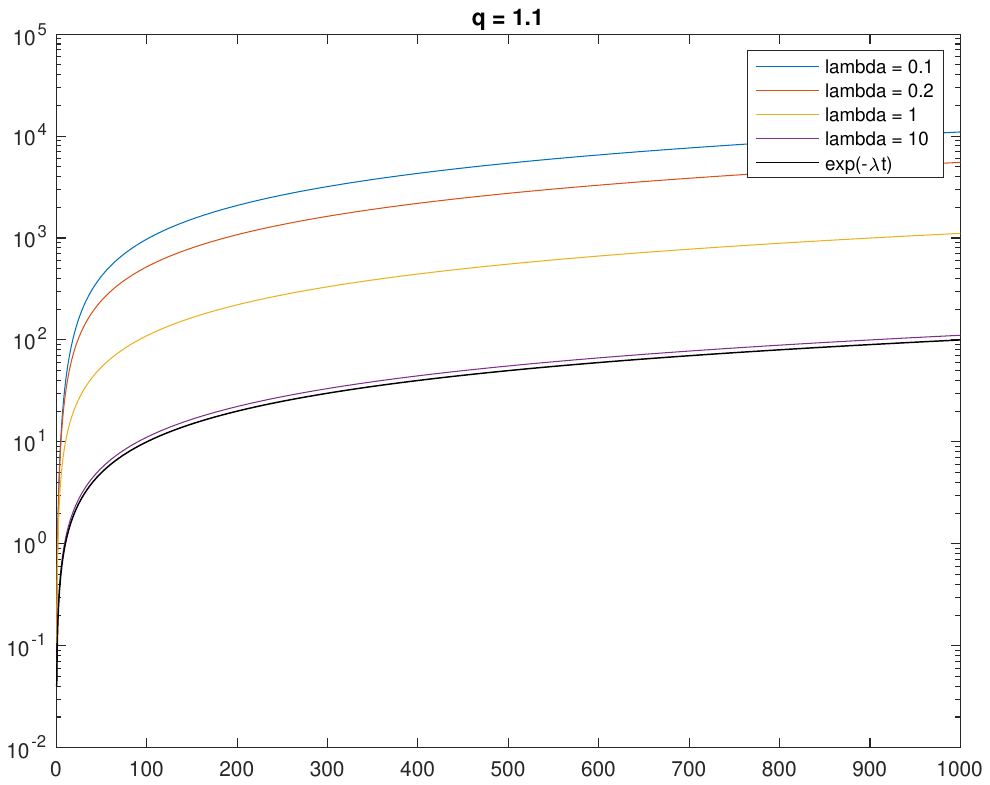}
    \end{subfigure}
    \caption[short]{Plots of $I_q(t)$ vs $\lambda$ for $q \approx 1$.}\label{fig:vary-lambda}
\end{figure}

\newpage
\bibliography{TsallisDIAbib}

\begin{thebibliography}{10}

\bibitem{Kraichnan-1961}
R.~H. Kraichnan, ``Dynamics of nonlinear stochastic systems,'' {\em Journal of Mathematical Physics}, vol.~2, no.~1, pp.~124--148, 1961.

\bibitem{Kraichnan-1965}
R.~H. Kraichnan, ``{Lagrangian-history closure approximation for turbulence},'' {\em The Physics of Fluids}, vol.~8, no.~4, pp.~575--598, 1965.

\bibitem{Kraichnan-1964}
R.~H. Kraichnan, ``{Kolmogorov's hypotheses and Eulerian turbulence theory},'' {\em The Physics of Fluids}, vol.~7, pp.~1723--1734, 1964.

\bibitem{MW-93}
C.~Y. Mou and P.~B. Weichman, ``Spherical model for turbulence,'' {\em Physical Review Letters}, vol.~70, pp.~1101--1104, 1993.

\bibitem{MW-95}
C.~Y. Mou and P.~B. Weichman, ``{Multicomponent turbulence, the spherical limit, and non-Kolmogorov spectra},'' {\em Physical Review E}, vol.~52, pp.~3738--3796, 1995.

\bibitem{Eyink}
G.~L. Eyink, ``{Large-N limit of the ``spherical model'' of turbulence},'' {\em Physical Review E}, vol.~49, pp.~3990--4002, 1994.

\bibitem{BKS}
B.~K. Shivamoggi, M.~Taylor, and S.~Kida, ``{On Some Mathematical Aspects of the Direct-Interaction Approximation in Turbulence Theory},'' {\em Journal of Mathematical Analysis and Applications}, vol.~229, no.~2, pp.~639--658, 1999.

\bibitem{BT}
B.~K. Shivamoggi and N.~Tuovila, ``{Direct interaction approximation for non-Markovianized stochastic models in the turbulence problem},'' {\em Chaos}, vol.~29, p.~063124, 2019.

\bibitem{UO}
G.~E. Uhlenbeck and L.~S. Ornstein, ``{On the Theory of the Brownian Motion},'' {\em The Physical Review}, vol.~36, pp.~823--841, 1930.

\bibitem{LGP}
I.~M. Lifshits, S.~A. Gredeskul, and L.~A. Pastur, {\em Introduction to the Theory of Disordered Systems}.
\newblock A Wiley Interscience publication, Wiley, 1988.

\bibitem{Weiss}
P.~R. Weiss, ``{The Application of the Bethe-Peierls Method to Ferromagnetism},'' {\em Phys. Rev.}, vol.~74, pp.~1493--1504, 1948.

\bibitem{Tsallis-intro}
C.~Tsallis, ``{Possible generalization of the Boltzmann-Gibbs statistics},'' {\em Journal of Statistical Physics}, vol.~52, pp.~479--487, 1988.

\bibitem{APS1989}
{L. C. Andrews, R. L. Phillips, B. K. Shivamoggi, J. K. Beck and M. L. Joshi}, ``{A statistical theory for the distribution of energy dissipation in intermittent turbulence},'' {\em Physics of Fluids A}, vol.~1, no.~6, pp.~999--1006, 1989.

\bibitem{AS1990}
L.~C. Andrews and B.~K. Shivamoggi, ``The gamma distribution as a model for temperature dissipation in intermittent turbulence,'' {\em Physics of Fluids A}, vol.~2, no.~1, pp.~105--110, 1990.

\bibitem{Doob}
J.~L. Doob, ``{The Brownian Movement and Stochastic Equations},'' in {\em {Selected Papers in Noise and Stochastic Processes}}, Dover, 1954.

\bibitem{WW}
G.~Wilk and Z.~Wlodarcyzk, ``{Interpretation of the non-extensivity parameter {$q$} in some applications of Tsallis statistics and L{\`e}vy distributions},'' {\em Physical Review Letters}, vol.~84, no.~13, pp.~2770--2773, 2000.

\bibitem{Beck}
C.~Beck, ``Dynamical foundations of non-extensive statistical mechanics,'' {\em The Physical Review Letters}, vol.~87, pp.~180601.1--.4, 2001.

\bibitem{Zwanzig}
R.~Zwanzig, {\em {Nonequilibrium Statistical Mechanics}}.
\newblock Oxford University Press, 2001.

\bibitem{Keller}
J.~B. Keller, ``Stochastic equations and wave propagation in random media,'' in {\em Stochastic Processes in Mathematical Physics and Engineering} (R.~Bellman, ed.), no.~16 in Proceedings of Symposia in Applied Mathematics, American Mathematical Society, 1964.

\bibitem{Kraichnan-1958}
R.~H. Kraichnan, ``Irreversible statistical mechanics of incompressible hydromagnetic turbulence,'' {\em The Physical Review}, vol.~109, pp.~1407--1422, 3 1958.

\bibitem{Kraichnan-1959}
R.~H. Kraichnan, ``{The structure of isotropic turbulence at very high Reynolds numbers},'' {\em Journal of Fluid Mechanics}, vol.~5, no.~4, p.~497–543, 1959.

\bibitem{Huang}
K.~Huang, {\em {Introduction to Statistical Physics}}.
\newblock CRC Press, {II}~ed., 2010.

\bibitem{TsallisReview}
C.~Tsallis, {\em Introduction to Nonextensive Statistical Mechanics}.
\newblock Springer, 2nd~ed., 2023.

\bibitem{Andrews}
L.~C. Andrews, {\em {Special Functions of Mathematics for Engineers}}.
\newblock SPIE Press, 1998.

\end{thebibliography}
\bibliographystyle{ieeetr}

\end{document}